\documentclass[conference] {IEEEtran}

\usepackage{amssymb}
\usepackage{fnpct}
\usepackage{amsmath}
\usepackage[pdftex]{graphicx}
\usepackage{epsfig}
\usepackage{framed}
\usepackage{etoolbox}
\usepackage{fnpct}
\usepackage{overpic}
\usepackage{url}
\usepackage[font=small,labelfont=bf]{caption} % Required for specifying captions to tables and figures
\usepackage{relsize} % Used for making text smaller in some places
\usepackage{setspace}
\usepackage{tikz}
\usepackage{pdfpages}
\usetikzlibrary{arrows,positioning}
\usepackage{enumitem}

\usepackage[utf8]{inputenc}

\newcommand{\rank}{{\rm rank}\, }

\newcommand\numberthis{\addtocounter{equation}{1}\tag{\theequation}}

\DeclareMathOperator*{\Spg}{spg}

\DeclareMathOperator*{\Circ}{circ}

\def\hide #1 {}
\long\def\longhide #1 {}

\usepackage{amsthm}

\theoremstyle{plain}
\newtheorem{theorem}{Theorem}[section]
\newtheorem*{theorem*}{Theorem}
\newtheorem*{maintheorem*}{Main Theorem}
\newtheorem{lemma}[theorem]{Lemma}

\newtheorem{claim}{Claim}
\newtheorem{corollary}[theorem]{Corollary}

\theoremstyle{definition}
\newtheorem{remark}[theorem]{Remark}
\newtheorem{assumption}{Assumption}

\newtheorem{definition}[theorem]{Definition}

\theoremstyle{definition}

\usepackage{mathtools}

\title{Injectivity of Multi-window Gabor Phase Retrieval}
\author{%
  Palina Salanevich\\Utrecht University\\ Email: p.salanevich@uu.nl
}

\begin{document}
\IEEEoverridecommandlockouts
\maketitle
\date{}

\begin{abstract}
 In many signal processing problems arising in practical applications, we wish to reconstruct an unknown signal from its phaseless measurements with respect to a frame. This inverse problem is known as the \emph{phase retrieval} problem. For each particular application, the set of relevant measurement frames is determined by the problem at hand, which motivates the study of phase retrieval for structured, application-relevant frames. In this paper, we focus on one class of such frames that appear naturally in diffraction imaging, ptychography, and audio processing, namely, \emph{multi-window Gabor frames}. We study the question of injectivity of the phase retrieval problem with these measurement frames in the finite-dimensional setup and propose an explicit construction of an infinite family of phase retrievable multi-window Gabor frames. We show that phase retrievability for the constructed frames can be achieved with a much smaller number of phaseless measurements compared to the previous results for this type of measurement frames. Additionally, we show that the sufficient for reconstruction number of phaseless measurements depends on the dimension of the signal space, and not on the ambient dimension of the problem.

\end{abstract}

%{\bf Key words: } phase retrieval, Gabor frames, expander graphs, angular synchronization, spectral clustering, order statistics.

\section{Introduction}\label{sec: intro}

Phase retrieval is the non-convex problem of signal reconstruction from the intensities of its (linear) measurements. It is motivated by a number of real-world applications within science and engineering. Among these applications are diffraction imaging \cite{Mill, Bunk} and ptychography \cite{Roden}, where the phases of the frame coefficients are lost in the measurement process; as well as audio processing \cite{rabiner1993fundamentals, balan1}, where phases may be too noisy to use them for reconstruction.

In the finite-dimensional case, the phase retrieval problem is formulated as follows. Let ${\Phi = \{\varphi_j\}_{j = 1}^N \subset \mathbb{C}^M}$ be a \emph{frame}, that is, a (possibly over-complete) spanning set of $\mathbb{C}^M$. We consider the phaseless measurement map ${\mathcal{A}_{\Phi}:\mathbb{C}^M \to \mathbb{R}^N}$ defined by ${\mathcal{A}_{\Phi}(x) = \{|\langle x, \varphi_j\rangle|^2\}_{j = 1}^N}$. The aim of the phase retrieval problem is to recover an unknown vector ${x\in \mathbb{C}^M}$ from its phaseless measurements $b = \mathcal{A}_\Phi(x)$. Since ${\mathcal{A}_{\Phi}(x) = \mathcal{A}_{\Phi}(e^{i\theta} x)}$ for any $\theta\in [0, 2\pi)$, the initial signal $x$ can be reconstructed up to a global phase factor at best. To factor out this ambiguity, we identify each $x\in \mathbb{C}^M$ with its up-to-a-global-phase equivalence class $[x] = \{e^{i\theta} x, ~ \theta\in [0, 2\pi)\}$ and consider the measurement map $\mathcal{A}_{\Phi}$ to be defined on the set of equivalence classes $\mathbb{C}^M /_\sim$.

Not for every frame $\Phi$ it is possible to uniquely reconstruct a signal $x$ from $\mathcal{A}_\Phi(x)$. The frames with injective associated phaseless measurement maps are called \emph{phase retrievable}. An important research directions in phase retrieval is to identify and describe classes of phase retrievable frames, see, e.g. \cite{grohs2020phase, balan1, Band, conca}. At the same time, in practical applications, measurement frames are often required to have a prescribed structure that is determined by the (physical) model behind the problem. For instance, measurement frames arising in diffraction imaging \cite{Mill, Bunk}, ptychography \cite{Roden}, and audio processing \cite{rabiner1993fundamentals, balan1} have a common structure of a \emph{(multi-window) Gabor frame} defined below. In this paper, we aim to address the following questions.
\begin{center}
\emph{How to construct phase retrievable multi-window Gabor frames of small cardinality? }
\end{center}
Our findings also provide a bound on the number of the phaseleless measurements with respect to a multi-window Gabor frame that is sufficient for reconstruction.

\begin{definition}
Let $G = \{g_r\}_{r=1}^R\subset \mathbb{C}^M$ be a set of \emph{windows} and $\Lambda\subset \mathbb{Z}_M\times \mathbb{Z}_M$. We define the \emph{multi-window Gabor frame} as the set of vectors $(G, \Lambda) = \{\pi (\lambda)g_r\}_{\lambda\in \Lambda, r\in\{1,\dots R\}}$, where 
\begin{itemize}
\item $\pi(k,\ell) = M_{\ell}T_k$ is a \emph{time-frequency shift operator};
\item $T_k x =  \left(x(m-k)\right)_{m\in \mathbb{Z}_M}$ is a \emph{translation operator};   
\item ${M_{\ell} x = \left(e^{2\pi i \ell m/M}x(m)\right)_{m\in \mathbb{Z}_M}}$ is a \emph{modulation operator}.
\end{itemize}
\noindent In the particular case when there is only one window $G = \{g\}$, the frame $(g, \Lambda)$ is called a \emph{Gabor~frame}.
\end{definition}

For Gabor frames, injectivity and stability results have been established only in the case when $\Lambda = \mathbb{Z}_M\times \mathbb{Z}_M$ \cite{bojarovska2016phase, alaifari2021stability, salanevich2019stability}. In particular,  \cite{bojarovska2016phase} provides a condition on the window $g$ that is sufficient for phase retrievability of the full Gabor frame $(g, \mathbb{Z}_M\times \mathbb{Z}_M)$. Reducing the cardinality of $\Lambda$ below $M^2$ is, however, a complicated task. Moreover, one can show that the phaseless measurement map $\mathcal{A}_{(g, \mathbb{Z}_M\times \mathbb{Z}_M)}$ lacks injectivity in the case when the window $g$ has short support or $x$ is allowed to have many consecutive zeros \cite{alaifari2019ill, alaifari2021stability}. 

A possible remedy for this problem is to simultaneously use several windows and consider phase retrieval with multi-window Gabor frames. In \cite{han2022quantum}, Han et.al. establish \emph{maximal span property} for a full multi-window Gabor frame $(G, \mathbb{Z}_M\times \mathbb{Z}_M)$, under the condition that ambiguity functions  of the windows in $G$ do not vanish simultaneously. As maximal span property implies phase retrievability of a frame, their result generalizes the condition obtained for (single-window) full Gabor frames in \cite{bojarovska2016phase}. In \cite{li2017phase}, Li et.al. consider frames $(G, T\times \mathbb{Z}_M)$ with $\vert T\vert = \frac{M}{L}$ and $R\geq L$, for a separation parameter $L$. They prove necessary and sufficient conditions for injectivity of $\mathcal{A}_{(G, \Lambda)}$, depending on the support size of the window $g$. 

Note that in both \cite{han2022quantum, li2017phase}, phase retrievability of a multi-window Gabor frame $(G, \Lambda)$ is established for $\vert (G, \Lambda) \vert = R \vert \Lambda \vert  = O(M^2)$.

\subsection{Main contribution}

In this paper, we manage to significantly reduce the number of measurements required to achieve injectivity of $\mathcal{A}_{(G, \Lambda)}$.

\begin{maintheorem*}
Let $C>3$ be a constant. Phase retrieval can be done on $\mathbb{C}^M$ from $CM(1+3\beta(M,C))$ multi-window Gabor frame phaseless measurements, where $\beta(M,C)$ is a measure of pseudorandomness defined in \eqref{eq: beta_definition} below.
\end{maintheorem*}

It follows from \cite[Lemmas~6~and~7]{mixon2} that %if $P\subset \mathbb{Z}_M$ is a random subset of cardinality $\vert P \vert  = O(\log M)$, then $\Vert P \Vert_u < \mathbf{P}(P)$ is satisfied with high probability. Thus, 
${\beta(M,C) \lesssim \log(M)}$, and thus phase retrieval can be done on $\mathbb{C}^M$ from at most $O(M\log(M))$ multi-window Gabor frame phaseless measurements, which is a significant improvement in comparison with $O(M^2)$. Furthermore, we show that, with a similar construction of the window set, phase retrieval can be done from $Cd(1+3\beta(d,C))$ multi-window Gabor frame phaseless measurements on any $d$-dimensional subspace of~$\mathbb{C}^M$.

In contrast with \cite{alaifari2021stability} and \cite{li2017phase}, where the proof of phase retrievability of (multi-window) Gabor frames relies on the properties of the ambiguity function of the window(s), we utilize the polarization idea of \cite{mixon1, mixon2}. We construct the set of windows so that the phaseless measurements corresponding to the auxiliary windows can be used to compute (relative) phases of the measurements corresponding to the primary window.

\subsection{Notation and definitions}
The following notation is used throughout the paper. 
\begin{itemize}[leftmargin=*]
\item $\mathbb{S}^{M-1} = \{ x\in \mathbb{C}^M \colon \Vert x \Vert_2 = 1\}$ is the unit sphere in~$\mathbb{C}^M$;
\item $x\odot y(m) = x(m)y(m)$ denotes the coordinatewise product of vectors $x, y\in \mathbb{C}^M$;
\item for a vector $b\in \mathbb{C}^k$, $\Circ(b) = \left(b | T_1 b | \dots | T_k b \right)$ denotes the circulant matrix whose columns are obtained by shifting vector $b$;
\item for a subset ${A\subset \{0,\dots, k-1\}}$, $\mathbf{1}_A$ denotes its characteristic function and $\mathbf{P}(A) = \vert A \vert / k$ denotes the \emph{density} of $A$.
\end{itemize}
Furthermore, the following definitions are used in the paper.
\begin{definition}
  We define the \emph{Fourier bias} of a set ${A\subset \{0,\dots, k-1\}}$ as $$\Vert A \Vert_u = \max_{m\neq 0}\vert \mathcal{F}(\mathbf{1}_A)(m)\vert.$$
  \end{definition}
  The Fourier bias of a set is a non-negative quantity which is equal to zero only for $A = \{0,\dots, k-1\}$ and $A = \emptyset$. It can get as large as the set density but is usually smaller \cite{tao}. Essentially, the Fourier bias of a set measures the maximal correlation of its indicator function with discrete harmonic functions. As for random sets this correlation is low with high probability, Fourier bias is used in additive combinatorics to measure pseudorandomness \cite{tao}. 
  
  In our construction, we are interested in small cardinality sets that have small Fourier bias. In particular, the cardinality of the constructed phase retrievable multi-window Gabor frame depends on the following quantity

  \begin{equation}\label{eq: beta_definition}
      \beta (M, C) = \min_{\substack{P\subset\mathbb{Z}_M\\P\ne\emptyset}} \left\lbrace \vert P \vert \colon \Vert P \Vert_u \leq \frac{C-3}{C-1}\mathbf{P}(P) \right \rbrace.
  \end{equation}

To construct the set of windows $G$ for a phase retrievable frame $(G, \Lambda)$, we employ some tools from algebraic graph theory.

\begin{definition}
For a $d$-regular graph $\mathcal{G}$ on $n$ vertices, let ${d = \lambda_0 \geq \lambda_1\geq \cdots \geq \lambda_n}$ denote the eigenvalues of its adjacency matrix. We define the \emph{spectral gap} of $\mathcal{G}$ as $$\Spg(\mathcal{G}) = 1 - \frac{1}{d}\max_{j\ne 0}\vert \lambda_j\vert.$$
\end{definition}

Clearly, a graph is disconnected if and only if its spectral gap is equal to 0. More generally, large $\Spg(G)$ ensures good connectivity properties of graph $G$ \cite{mixon1, harsha}.

\bigskip

The remaining part of this paper is organized as follows. In Section~\ref{sec: phase retrievable full dim} we describe the construction of the window set, and prove phase retrievability of the respective multi-window Gabor frame, under certain assumptions on the primary window. In Section~\ref{sec: phase retrievable low dim priors}, we generalize the results of Section~\ref{sec: phase retrievable full dim} to show that the sufficient number of measurements with respect to the constructed multi-window Gabor frame depends on the dimension of the signal space rather than on the ambient dimension of the problem.  We conclude the paper with a brief discussion of the future research directions in Section~\ref{sec: discussion}.

\section{Phase retrievable multi-window Gabor frames}\label{sec: phase retrievable full dim}

In this paper, we propose a construction of the set of windows $G$, such that the corresponding multi-window Gabor frame has injective associated phaseless measurement map $\mathcal{A}_{(G, \Lambda)}$. Our construction is inspired by the idea of the polarization algorithm \cite{mixon1, mixon2, pfander2019robust}. Let us consider the set of windows $G = \{g\}\cup G'$, where we distinguish a \emph{primary} window $g$ and call the rest of the windows in $G'$ \emph{auxiliary}. We construct auxiliary windows so that phaseless measurements of a signal with respect to $(G', \Lambda)$ can be used to compute relative phases between (some of) the phaseless measurements with respect to $(g, \Lambda)$. More precisely, 

\begin{align*}
    G' = \{ g_{qpt} = g\odot s_{qpt}\}_{\substack{q \in Q,~p\in P \\t\in \{0, 1, 2\}}}, \text{ where} \numberthis\label{eq: auxiliar masks}\\
    s_{qpt} (m) = 1 + e^{2\pi i \left( mp/M + t/3\right)}\frac{g(m-q)}{g(m)}.
\end{align*}

\begin{lemma}\label{lem: relative phases}
Let $(G, \Lambda)$ be a multi-window Gabor frame with the set of windows $G = \{g\}\cup G'$, where $G'$ is defined as in~\eqref{eq: auxiliar masks}. Then, for any $(k,\ell)\in \Lambda$, $q\in Q$, and $p\in P$,
\begin{align*}
    \langle x, \pi(k,\ell)g \rangle \overline{\langle x, \pi(k+q, \ell + p)g \rangle} = \\ \frac{e^{2\pi i k p/M}}{3}\sum_{t = 0}^{2} e^{2\pi i t/3} \vert \langle x, \pi(k, \ell)g_{qpt}\rangle \vert^2
\end{align*}
\end{lemma}
\begin{proof}
First, let us observe that by definition of $g_{qpt}$,
\begin{align*}
    & \langle x, \pi(k,\ell)g_{qpt}\rangle = \sum_{m\in\mathbb{Z}_M}x(m)e^{\frac{-2\pi i \ell m}{M}}\overline{g(m-k)s_{qpt}(m-k)} \\ 
    & = \sum_{m\in\mathbb{Z}_M}x(m)e^{-2\pi i \ell m/M}\overline{g(m-k)} \\ & + \sum_{m\in\mathbb{Z}_M}x(m)e^{-2\pi i \left(\frac{\ell (m+p)}{M} + \frac{t}{3}\right)}e^{\frac{2\pi i kp}{M}} \overline{g(m - (k+q)}\\
    & = \langle x, \pi(k, \ell)g\rangle + e^{\frac{-2\pi i t}{3}}e^{\frac{2\pi i kp}{M}} \langle x, \pi(k+q, \ell+p)g\rangle.
\end{align*}
By applying the polarization identity
\begin{align*}
    a\overline{b} = \frac{1}{3}\sum_{t = 0}^2 e^{\frac{2\pi i t}{3}} \vert a + e^{\frac{-2\pi i t}{3}}b\vert^2,~~a,b\in\mathbb{C} 
\end{align*}
with $a = \langle x, \pi(k, \ell)g\rangle$ and $b = e^{\frac{2\pi i kp}{M}} \langle x, \pi(k+q, \ell+p)g\rangle$, we obtain the desired equality.
\end{proof}

To show that the multi-window Gabor frame $(G, \Lambda)$ constructed above is phase retrievable, we need additional assumptions on the primary window $g$ and index sets $P$ and~$Q$.

\begin{assumption}\label{assumption: full spark}Window $g\in \mathbb{C}^M$ is nowhere vanishing, such that the corresponding Gabor frame $(g, \Lambda)$ is \emph{full-spark}, that is, any $M$ vectors in $(g, \Lambda)$ are linearly independent.
\end{assumption}

Note that the set of $g\in\mathbb{S}^{M-1}$ for which Assumption~\ref{assumption: full spark} is satisfied is a full measure set in $\mathbb{S}^{M-1}$ \cite{pfander1, malik}. In particular, if ${g\sim \text{Unif.} \left( \mathbb{S}^{M-1}\right)}$, then Assumption~\ref{assumption: full spark} is satisfied with probability~$1$.

\begin{assumption}\label{assumption: small Fourier bias}  A subset $Q\subset \mathbb{Z}_M$ satisfies
${\Vert Q\Vert _u \leq c \mathbf{P}(Q)}$, for some constant $c\in(0,1) $.
\end{assumption}

Note that for any subset $Q\subset \mathbb{Z}_M$ we have
${\Vert Q\Vert _u \leq \mathbf{P}(Q)}$, and equality holds only for sets $Q$ with very specific structure (namely, for cosets of a proper subgroup of $\mathbb{Z}_M$)~\cite{tao}. Subsets that satisfy Assumption~\ref{assumption: small Fourier bias} should have small Fourier bias. Such subsets are called \emph{linearly uniform} or \emph{pseudo-random}. In particular, if we generate a subset $Q$ at random, by uniformly and independently selecting elements of $\mathbb{Z}_M$ with probability $\frac{c^2\log(M)}{9M}$, then Assumption~\ref{assumption: small Fourier bias} is satisfied with high probability~\cite{tao,mixon2}.

\medskip

We formulate our result as follows.

\begin{theorem}\label{thm: phase retrievalble full dim}
Let $g\in\mathbb{C}^M$ satisfy Assumption~\ref{assumption: full spark} and ${\Lambda = T\times F\subset \mathbb{Z}_M\times \mathbb{Z}_M}$ with $\vert \Lambda\vert > CM$, for some $C>3$. Suppose further that sets $Q\subset T-T$ and $P\subset F-F$ satisfy Assumption~\ref{assumption: small Fourier bias} with $c = \frac{C-3}{C-1}$. Then $(G, \Lambda)$ with $G = \{g\}\cup G'$ defined as in \eqref{eq: auxiliar masks} is a phase retrievable frame.
\end{theorem}
\begin{proof}
Let us consider a graph $(\Lambda, E)$ with the set of vertices $\Lambda$ and the set of edges 
\begin{equation}\label{eq: edge set}
    E = \{\left( (k, \ell), (k', \ell') \right) \colon k' - k\in Q, \ell' - \ell\in P\}\subset \Lambda \times \Lambda.
\end{equation} 
Then, for any edge $e = \left((k, \ell), (k', \ell')\right)\in E$, such that $\vert \langle x, \pi(k, \ell)g\rangle \vert\neq 0$ and $\vert \langle x, \pi(k', \ell')g\rangle \vert\neq 0$, using Lemma~\ref{lem: relative phases} we obtain that the relative phase  $$\omega_e = \frac{\langle x, \pi(k, \ell)g \rangle}{\vert\langle x, \pi(k, \ell)g \rangle\vert}\left( \frac{\langle x, \pi(k', \ell')g \rangle}{\vert\langle x, \pi(k', \ell')g \rangle\vert} \right)^{-1}$$ can be computed from phaseless measurements $\mathcal{A}_{(G', \Lambda)}$ as
$$\frac{e^{\frac{2\pi i k p}{M}}}{3\vert\langle x, \pi(k, \ell)g \rangle\vert \vert\langle x, \pi(k', \ell')g \rangle\vert}\sum_{t = 0}^{2} e^{\frac{2\pi i t}{3}} \vert \langle x, \pi(k, \ell)g_{qpt}\rangle \vert^2,$$ where $p = \ell' - \ell$ and $q = k' - k$.

We are going to use the obtained graph $(\Lambda, E)$ with weighted edges to reconstruct (up to a global phase shift) the phases of (a subset of) the frame coefficients of $x$ with respect to the Gabor frame $(g, \Lambda)$.
Note that for any $(k,\ell)\in \Lambda$, such that $\vert \langle x, \pi(k, \ell)g\rangle \vert = 0$, the relative phase $\omega_e$ is not defined for any $e = \left((k, \ell), (k', \ell')\right)\in E$, thus we delete these edges from the graph to obtain a modified graphs $(\Lambda, E')$ with the weighted edges, where $$E' = E\setminus \{(\pi, \pi')\colon \vert \langle x, \pi(\pi)g\rangle \vert = 0 \text{ or } \vert \langle x, \pi(\pi')g\rangle \vert = 0\}.$$ 

\begin{claim}\label{claim: connected component}
The graph $(\Lambda, E')$ constructed above has a connected component of size at lest $M$.
\end{claim}
\begin{proof}[Proof of Claim~\ref{claim: connected component}]
Let $A$ be the adjacency matrix of the graph $(\Lambda, E)$. By construction of $E$, $A = \Circ\left(\mathbf{1}_Q \right)\otimes \Circ\left(\mathbf{1}_P \right)$,
where $\otimes$ denotes the Kronecker product. Then, the eigenvalues of $A$ are given by 
\begin{align*}
    \lambda_{jj'}(A) & =  \lambda_j\left(\Circ\left(\mathbf{1}_Q \right)\right) \lambda_{j'}\left(\Circ\left(\mathbf{1}_P \right)\right) \\
    & = \sum_{m\in\mathbb{Z}_M} \mathbf{1}_Q(m)e^{\frac{-2\pi i j m}{M}} \sum_{m'\in\mathbb{Z}_M} \mathbf{1}_P(m')e^{\frac{-2\pi i j' m'}{M}},
\end{align*}
as the eigenvalues of a circulant matrix $\Circ\left(\mathbf{1}_Q \right)$ are given by the entries of the Fourier transform $M\mathcal{F}\left(\mathbf{1}_Q \right)$. Since
\begin{align*}
    \vert \lambda_{j}(\Circ\left(\mathbf{1}_Q \right))\vert & \leq \sum_{m\in\mathbf{Z}_M} \left\vert \mathbf{1}_Q(m)e^{\frac{-2\pi i j m}{M}} \right\vert = \vert Q \vert,
\end{align*}
with equality when $j = 0$, it follows that $$\lambda_{\max}(\Circ\left(\mathbf{1}_Q \right)) = \lambda_{0}(\Circ\left(\mathbf{1}_Q \right)) = \vert Q \vert.$$ Similarly, $\lambda_{\max}(\Circ\left(\mathbf{1}_P \right)) = \lambda_{0}(\Circ\left(\mathbf{1}_P \right)) = \vert P \vert$, and $\lambda_{\max}(A) = \lambda_{00}(A) =  \vert Q \vert \vert P \vert$. Using this and the definition of the Fourier bias $\Vert \cdot \Vert_u$ of a set, we get
\begin{align*}
    & \Spg(\Lambda, E) = 1 - \frac{1}{\vert Q \vert \vert P \vert} \max_{(j,j')\ne (0,0)}\vert\lambda_{jj'}(A)\vert \\ & = 1 - \frac{1}{\vert Q \vert \vert P \vert} \max_{(j,j')\ne (0,0)}\vert\lambda_{j}(\Circ\left(\mathbf{1}_Q \right))\vert\vert\lambda_{j'}(\Circ\left(\mathbf{1}_P \right))\vert \\
    & = 1 - \max\left\lbrace\frac{M}{\vert Q \vert} \Vert Q \Vert_u, \frac{M}{ \vert P \vert} \Vert P \Vert_u \right\rbrace,
\end{align*}
that is, as both $P$ and $Q$ satisfy  Assumption~\ref{assumption: small Fourier bias} and $\vert \Lambda \vert >CM$, $$\Spg (\Lambda, E) = 1 - \max\left\lbrace\frac{\Vert Q \Vert_u}{\mathbf{P}(Q)} , \frac{\Vert P \Vert_u}{ \mathbf{P}(P)}  \right\rbrace \geq \frac{2}{C-1} \geq \frac{2M}{\vert \Lambda \vert - M}.$$
The graph $(\Lambda, E')$ is obtained from $(\Lambda, E)$ by removing $$k = \left\vert \{(\lambda, \lambda')\colon \vert \langle x, \pi(\lambda)g\rangle \vert = 0 \text{ or } \vert \langle x, \pi(\lambda')g\rangle \vert = 0\}\right\vert$$ edges. By Assumption~\ref{assumption: full spark}, $(g, \Lambda)$ is a full spark frame, thus $\left\vert \{\pi\in \Lambda \colon \vert \langle x, \pi(\lambda)g\rangle \vert = 0 \}\right\vert\leq M-1$ for any $x\neq 0$, and ${k \leq |P||Q|(M-1)}$. Applying \cite[Lemma 5.2]{harsha}, we obtain that $(\Lambda, E')$ has a connected component of size at least $\left( 1 - \frac{2M}{\vert \Lambda \vert \Spg(\Lambda, E)}\right)\vert \Lambda \vert = M$.
\end{proof}

Using Claim~\ref{claim: connected component}, let us fix $(\Lambda', E'')$ to be a connected component of $(\Lambda, E')$ with $\vert \Lambda'\vert \ge M$ and $E'' = E'\cap \Lambda' \times \Lambda'$. By Assumption~\ref{assumption: full spark}, $(g, \Lambda)$ is a full spark frame, thus any signal $x\in \mathbb{C}^M$ can be reconstructed from the set of its frame coefficients $\left\lbrace\langle x, \pi(\lambda)g \rangle = \frac{\langle x, \pi(\lambda)g \rangle}{\vert\langle x, \pi(\lambda)g \rangle\vert} \sqrt{\vert\langle x, \pi(\lambda)g \rangle\vert^2}\right\rbrace_{\lambda\in \Lambda'}$. For this reason, to uniquely recover $x$, it is enough to determine (up to a global phase shift) the phases of the frame coefficients $\frac{\langle x, \pi(\lambda)g \rangle}{\vert\langle x, \pi(\lambda)g \rangle\vert}$, for all $\lambda\in \Lambda'$. To do so, we iteratively propagate relative phases $\omega_e$, $e\in E''$ inside the connected component $(\Lambda', E'')$ or apply angular synchronization algorithm~\cite{singer}. 
\end{proof}

The Main Theorem from Section~\ref{sec: intro} can be deduced from Theorem~\ref{thm: phase retrievalble full dim} by choosing $g\sim \text{Unif.} \left( \mathbb{S}^{M-1}\right)$, $\Lambda = T\times \mathbb{Z}_M$ with $\vert T \vert = C$, $Q = T-T$, and $P$ being a minimizer in~\eqref{eq: beta_definition}.

\begin{remark}
Note that the proof of Theorem~\ref{thm: phase retrievalble full dim} does not only show that under Assumptions~\ref{assumption: full spark}~and~\ref{assumption: small Fourier bias} the multi-window Gabor frame is phase retrievable, but also suggests a reconstruction algorithm that is similar to~ \cite{mixon2, pfander2019robust}.
\end{remark}

The number of vectors in $(\{g\}\cup G', \Lambda)$ is $\vert \Lambda \vert (1 + 3\vert Q\vert \vert P \vert) = O(\vert Q\vert \vert P \vert M)$. To reduce the cardinality the frame we constructed, we would like to be able to construct small subsets $P,Q\subset \mathbb{Z}_M$ with small Fourier bias. In particular, for random subset $P\subset \mathbb{Z}_M$ of cardinality $\vert P \vert  = O(\log M)$ it has been shown in \cite[Lemmas~6~and~7]{mixon2} that $\Vert P \Vert_u < c\mathbf{P}(P)$ for some $c\in (0,1)$ with high probability. Using this observation, we deduce the following corollary also proven in \cite{pfander2019robust}.

\begin{corollary}[Theorem 3.4, \cite{pfander2019robust}]\label{cor: random sets}
Let $g\sim \text{Unif.} \left( \mathbb{S}^{M-1}\right)$ and $\Lambda = T\times \mathbb{Z}_M$ with $\vert T \vert = C$. Suppose further that $Q = T-T$ and $P$ is a random subset of $\mathbb{Z}_M$, such that ${\mathbf{1}_P(m)\sim \text{i.i.d. Bernoulli}\left(\frac{\alpha \log(M)}{M}\right)}$. Then, with high probability, $(G, \Lambda)$ with $G = \{g\}\cup G'$ defined as in \eqref{eq: auxiliar masks} is a phase retrievable frame.
\end{corollary}

\section{Multi-window Gabor phase retrieval under lower-dimensional priors}\label{sec: phase retrievable low dim priors}

In this section, we generalize findings of Theorem~\ref{thm: phase retrievalble full dim} to the case when there is some prior knowledge available on the signal of interest $x$. More precisely, we study how the number of phaseless multi-window Gabor measurements sufficient for reconstruction of $x$ changes in the case when $x$ is an element of an (unknown) lower-dimensional subspace of $\mathbb{C}^M$.

\begin{theorem}\label{thm: phase retrievalble subspace}
Let $g\in\mathbb{C}^M$ satisfy Assumption~\ref{assumption: full spark} and ${\Lambda = T\times F\subset \mathbb{Z}_M\times \mathbb{Z}_M}$ with $\vert \Lambda\vert > Cd$, for some $C>3$. Suppose further that sets $Q\subset T-T$ and $P\subset F-F$ satisfy Assumption~\ref{assumption: small Fourier bias} with $c = \frac{C-3}{C-1}$. Then, for any ${W\in \mathbb{C}^{d\times M}}$ with $d\leq M$ and $\rank(W) = d$, the phaseless map $\mathcal{A}_{(G, \Lambda)}$  with $G = \{g\}\cup G'$ defined as in \eqref{eq: auxiliar masks} is injective on ${\{x\in \mathbb{C}^M \colon x = Wh,~ h\in \mathbb{C}^d\}}$.
\end{theorem}

\begin{proof}
    First, note that for $x= Wh$, we have $${\mathcal{A}_{(G, \Lambda)}(x) = \mathcal{A}_{\Psi}(h)},$$ where $\Psi = \{W^*\varphi \colon \varphi\in (G, \Lambda)\}$.
    As $x$ is uniquely determined by $h$, it is enough to show that $h$ can be uniquely (up to a global phase factor) recovered  from its phaseless measurements $\mathcal{A}_{\Psi}(h)$. 

    Let us write $\Psi = \Psi_g\cup \Psi_{G'}$, where $\Psi_g = \{W^*\pi(\lambda)g\}_{\lambda\in \Lambda}$ and $\Psi_{G'} = \{W^*\pi(\lambda)g_{qpt}\}_{\lambda\in \Lambda, g_{qpt}\in G'}$. Similar to the proof of Theorem~\ref{thm: phase retrievalble full dim}, we are going to use $\Psi_{G'}$ to compute relative phases between the frame coefficients of $h$ with respect  to $\Psi_{g}$. Indeed, following the proof of Lemma~\ref{lem: relative phases}, we observe that for $\lambda = (k,\ell)$ and $\lambda' = (k+q, \ell+p)$
    \begin{align*}
        W^*\pi(\lambda)g_{qpt} & = W^*\left(\pi(\lambda)g + e^{\frac{-2\pi i t}{3}}e^{\frac{2\pi i kp}{M}} \pi(\lambda')g \right)\\ & = W^*\pi(\lambda)g + e^{\frac{-2\pi i t}{3}}e^{\frac{2\pi i kp}{M}} W^*\pi(\lambda')g.
    \end{align*}
    Thus, Lemma~\ref{lem: relative phases} can be used to show that phaseless measurements of $h$ with respect to $\Psi_{G'}$ allow us to compute relative phases $\omega_e$ for all $e = \left( \lambda, \lambda'\right) \in E$ with $\vert \langle h, W^*\pi(\lambda)g \rangle\vert \neq 0$ and $\vert \langle h, W^*\pi(\lambda')g \rangle\vert \neq 0$, where $E$ is defined as in \eqref{eq: edge set}. 
    
    \begin{claim}\label{claim: full spark low dim projection}
        $\Psi_g = \{W^*\pi(\lambda)g\}_{\lambda\in \Lambda}\subset \mathbb{C}^d$ is a full spark frame.
    \end{claim}
    \begin{proof}[Proof of Claim~\ref{claim: full spark low dim projection}]
        By Assumption~\ref{assumption: full spark}, $(g, \Lambda)$ is a full spark frame, that is, for any distinct $\lambda_1, \dots, \lambda_k\in \Lambda$, $$\rank \Big( \pi(\lambda_1)g, \dots, \pi(\lambda_k)g \Big) = \min\{k, M\}.$$ Since $\rank(W) = d$, it follows that $$\rank \Big( W^*\pi(\lambda_1)g, \dots, W^*\pi(\lambda_k)g \Big) = \min\{d, k, M\}.$$ Thus, vectors $W^*\pi(\lambda_1)g, \dots, W^*\pi(\lambda_k)g$ are linearly independent if $k\leq d$.
    \end{proof}
    Let us consider the graph $(\Lambda, E)$. From Claim~\ref{claim: full spark low dim projection}, it follows that $\left\vert \{\lambda\in \Lambda \colon \vert \langle h, W^*\pi(\lambda)g \rangle\vert = 0\}\right\vert\leq d-1$, thus the number of edges $e\in E$, for which $\omega_e$ is not defined is at most $\vert P\vert \vert Q \vert (d-1)$. Applying Claim~\ref{claim: connected component} with $d$ in place of $M$ and \cite[Lemma 5.2]{harsha}, we derive that deleting these edges from the graph leads to a connected component $(\Lambda', E'')$ of size $\vert \lambda'\vert \geq d$. From Claim~\ref{claim: full spark low dim projection} it follows that $h$ can be recovered from its frame coefficients with respect to $\{W^*\pi(\lambda)g \}_{\lambda\in \Lambda'}$. The proof is then concluded by computing the phases of the frame coefficients $\{\langle h, W^*\pi(\lambda)g\rangle \}_{\lambda\in \Lambda'}$ from the relative phases $\omega_e$, $e\in E''$ using phase propagation or angular synchronization.
\end{proof}

Note that, similarly to Corollary~\ref{cor: random sets}, by selecting the window $g$ and the sets $P$ and $Q$ at random, one can derive that for any $W\in \mathbb{C}^{d\times M}$ the phaseless map $\mathcal{A}_{(G, \Lambda)}$ with ${\vert (G, \Lambda)\vert  = O(d\log(d))}$ is injective on ${\{x\in \mathbb{C}^M \colon x = Wh,~ h\in \mathbb{C}^d\}}$ with high probability. That is, in the case when $x$ is known to be an element of a lowed-dimensional subspace, the ambient dimension $M$ can be replaced in the sufficient number of phaseless measurements with the subspace dimension $d$.

\section{Discussion}\label{sec: discussion}

In this paper, we showed how polarization idea \cite{mixon1} can be used to construct phase retrievable multi-window Gabor frames of small cardinality. As the construction of such frames relies on the small subsets of $\mathbb{Z}_M$ with small Fourier bias, explicit construction of such subsets for every $M$ can further reduce the number of the phaseless measurements required for the signal reconstruction. 

In Section~\ref{sec: phase retrievable low dim priors}, we discussed multi-window Gabor phase retrieval under the assumption that the set of signals we aim to recover lies in a lower-dimensional subspace of $\mathbb{C}^M$. Surely, this kind of priors is not general enough to be used in practice. A more interesting, both from mathematical and practical points of view, class of priors are \emph{generative priors} studied in~\cite{hand2018phase}. There we assume that $x = W(h)$, $h\in \mathbb{C}^d$, for some non-linear generative map $W$ given, for instance,  by a neural network. Studying phase retrievability of the multi-window Gabor frames and determining how the sufficient number of the phaseless measurements changes under such  generative priors is an important direction for further research. 

\section*{Acknowledgements}
Palina Salanevich is supported by NWO Talent programme Veni ENW grant, file number VI.Veni.212.176.

%\printbibliography
\let\Section\section 
\def\section*#1{\Section{#1}} 
\bibliographystyle{alpha}
\bibliography{SampTA_Gabor_phase_retrieval}

\end{document}